\documentclass{llncs}
\usepackage{graphicx}
\usepackage[greek,english]{babel}
\usepackage[iso-8859-7]{inputenc}
\usepackage{subfig}
\usepackage{caption}
\usepackage{subfloat}

\usepackage{enumerate}

\usepackage{eurosym}

\begin{document}
\bibliographystyle{plain}

\latintext
\title{$(\alpha,\beta)$ Fibonacci Search}
\author{Pavlos S. Efraimidis}

\institute{Department of Electrical and Computer Engineering,\\
Democritus University of Thrace, Building A\\
University Campus, 67100 Xanthi, Greece \\
pefraimi@ee.duth.gr}
% \\
%\vspace{0.2cm} \today}

\maketitle

\begin{abstract}
Knuth~\cite[Page 417]{Kn81} states that ``the (program of the) Fibonaccian search technique looks very mysterious
at first glance'' and that ``it seems to work by magic''. In this work, we show that there
is even more magic in Fibonaccian (or else Fibonacci) search. We present a generalized Fibonacci
procedure that follows perfectly the implicit optimal decision tree for search problems
where the cost of each comparison depends on its outcome.
\end{abstract}

\section{Introduction}
\label{sec:intro}
Binary search is the most common algorithm for comparison-based search
in sorted arrays in which random access is permitted. The worst-case
complexity is logarithmic on the length $n$ of the sequence.
A general information-theoretic lower bound based on decision trees shows
that any comparison-based search algorithm needs in the worst-case at least
logarithmic steps to locate an item in a sorted array.

In the analysis of binary search, the cost of each comparison step
is usually some fixed constant. However, there are natural problems where the
cost of a comparison depends on its outcome.
For example, during the search an overestimation of the unknown value
may cost more than an underestimation or the opposite.
An interesting example is provided by TCP (Transmission Control Protocol) flows
that transmit data over the Internet. Each TCP flow has a parameter
called the congestion window, which is used to regulate its transmission rate.
A small congestion window will cause the flow to use less than the available bandwidth whereas
a large congestion window will cause delays and packet losses in case of
overflows. Most TCP flows use the AIMD (Additive Increase
Multiplicative Decrease) algorithm to adjust their congestion window
to the current network conditions~\cite{KKPS00,ET08}.
The cost induced to a TCP flow for underestimating the optimal congestion window
differs from the cost of overestimating it.

We consider the following toy-example of asymmetric cost binary search.
Assume that we want to experimentally estimate the minimum speed at
which the airbags of a particular car will deploy.
%We are aware that
%there are many parameters and inputs from sensors that influence
%the decision to deploy the airbags in case of a crash.
Assuming that all other parameters of the experiment are held constant
(ceteris paribus assumption),
we want to identify the minimum speed with one decimal digit precision,
at which the airbags deploy. We are given the information that the requested
minimum speed is in the closed interval 10-20 km/h (interval of uncertainty).
Finally, there is asymmetry in the cost of each comparison (crash test).
The damage caused by a crash test in the interval of uncertainty is {500\euro}
if the airbags do not deploy and {1500\euro} if the airbags open. What is the
minimum worst-case cost to identify the requested speed? The solution of this
example is given in Section~\ref{sec:disc}. A similar example is the leaky shower
problem described in~\cite{KR89}.

An amusing example of a closely related search problem is given in the lecture notes
of~\cite{Ra01}. Given two absolutely identical eggs, how many times do we have to
drop an egg to identify how strong the eggs are? In this problem, the unknown value
has to be found with at most two overestimates. The cracking eggs problem is
an example of asymmetric cost binary search where there is a fixed bound on the number
of the outcomes of one type, usually the overestimates. Binary search with a bounded
number of overestimates is discussed in~\cite{BB82,vLSUZ87}, where it is shown
that it is related to binomial trees.

In this work, we consider the problem of searching in a sorted array when
the cost of each comparison depends on its outcome. For example,
in the binary case, each comparison with outcome ``$\leq$'' has a fixed
cost $\alpha$ and each comparison with outcome ``$>$'' a fixed cost $\beta$.
We assume that $\alpha$ and $\beta$ are integers or rational multiples of each other,
such that $\alpha \geq 1$ and $\beta \geq 1$. We then generalize the results
to the case of searching with more than two outcomes.
A problem related to the (binary) search procedure is the construction
of (binary) search trees.
An $O(n)$ time algorithm to construct a binary search
tree for $(\alpha,\beta)$ binary search is presented in~\cite{CW77}.
The results of~\cite{CW77} are based on an interesting relation of the search tree structure
to the Pascal's triangle and are valid for any positive numbers $\beta \geq \alpha > 0$.
The case of $(1,2)$ binary search trees is further discussed in~\cite{ORSW84}
where the relation of the $(1,2)$ binary search problem to the original
Fibonacci sequence is identified. Binary trees with choosable edge lengths
are discussed in~\cite{MR09}.
A general treatment of binary search with asymmetric costs is given in~\cite{KR89}.
The costs of binary search in~\cite{KR89} can be real numbers, but the search
procedure proposed is only near optimal. The same problem is studied in~\cite{Hi90}
with tools from dynamic programming. Tree structures for searching with asymmetric
costs when the number of outcomes of each comparison may be larger than two is
discussed in~\cite{CG01} where optimal lopsided trees are analyzed.

The focus of this work is on a lightweight search algorithm for asymmetric cost search
and not on the explicit construction of the corresponding search tree.
As such, the main contribution of this work is an optimal lightweight search procedure
for $(\alpha,\beta)$ binary search without the explicit construction of the binary
search tree. Furthermore, we show that the results of this work are valid for search
trees of higher degrees.
The generalization of the Fibonacci sequences used in this work seems to naturally
fit the search problem and gives simple, optimal results.
Even though the problem has been studied in the past, no algorithm
that is both optimal with respect to its running time and the cost of
the generated solution has been presented.

% for the asymmetric cost searching with
%weights that are integer or rationally related
%\footnote{Searching
%is one of the earliest problems that has been studied in the field of Algorithms.
%The literature (or at least part of it) of the early days of Computer Science is
%sometimes not easy to find and access. After thorough research,
%to the best of our knowledge, the problem discussed in this paper
%has not been addressed before.}.

%Furthermore, we assume that $\alpha \neq \beta$, because otherwise the problem degenerates to
%normal binary search. XXXX However, our algorithms and their analysis
%are still valid (Are they ?).

%In our results we will encounter a generalization of Fibonacci sequences.
Interestingly, there is an early search algorithm called Fibonacci search
or Fibonaccian search~\cite{Fe60,Kn81}, which uses the classic Fibonacci
sequence to lead the item selection procedure in the sorted array\footnote{The term Fibonacci search is also used
in the literature for a technique that locates the minimum of a unimodal function
in a given interval. See for example~\cite{Av03}. This interpretation of the
term Fibonacci search is not related to our work.}. Fibonacci search
is a variation of binary search that avoids the division operations
required in binary search to select the middle item. At that early days
of computer science a division by two or even a simple shift operation could be
an undesirably expensive operation for a computer program.

Knuth~\cite[Page 417]{Kn81} states that ``the (program of the) Fibonaccian search technique looks very mysterious
at first glance'' and that ``it seems to work by magic''. In this work, we show that there
is even more magic in Fibonaccian (or else Fibonacci) search. The generalized Fibonacci
procedure of this work follows perfectly the implicit optimal decision tree for search problems
where the cost of each comparison depends on its outcome.

{\bf Contribution.}
We present $(\alpha,\beta)$ Fibonacci search, an optimal lightweight search procedure for
searching with asymmetric costs. The algorithm is simple and intuitive and the first, to our
knowledge, optimal algorithm for a classic problem when the comparison costs are fixed integers or
rationally multiples of each other.
We define Fibonacci decision trees and use them to prove the information theoretic lower
bounds of the search problem. Even though this structure exists implicitly in
more general works like~\cite{CG01}, its significance for search and decision problems has not been
sufficiently stressed so far.
Moreover, we present and solve an interesting guessing game with $h$ outcomes per comparison
and discuss applications of it in prefix codes. Finally, we present an abstract class definition
and an implementation of the search algorithm.

The rest of this work is organized as follows: The asymmetric cost search problem is
defined in Section~\ref{sec:abs}. In Section~\ref{sec:Sequences},
Fibonacci-like sequences are defined. The lower bound on the worst-case cost is
presented in Section~\ref{sec:bound}. The $(\alpha,\beta)$ Fibonacci search algorithm is
described in Section~\ref{sec:alg}. Applications of $(\alpha,\beta)$ Fibonacci search
are presented in Section~\ref{sec:Apps} and a final discussion is given in Section~\ref{sec:disc}.

\section{$(\alpha,\beta)$ Binary Search}
\label{sec:abs}
The formal definition of the asymmetric binary search problem follows.
\begin{definition}
The $(\alpha,\beta)$ binary search problem. An array $V$ of $n$ items sorted in increasing order is given. Each item of
the array can be accessed in time $O(1)$. For any value $x$ and any item $v_k$ of $V$
the cost of the comparison $x \leq v_k$ is
\begin{center}
\begin{math}
\mbox{cost of } (x \leq v_k) = \left\{
\begin{array}{ll}
\alpha, \; \mbox{ if the outcome if true, } \\
\beta, \; \mbox{ if the outcome is false. }
\end{array}
\right .
\end{math}
\end{center}
Parameters $\alpha$ and $\beta$ are strictly positive integers.
The cost of a search is equal to the sum of the costs of all comparison
operations $x \leq v_k$ performed during the search.

We assume that the searched item $x$ is in the array $V.$ The results of this work
can be extended to searches for items that may not belong to $V$ by comparing at the end
of the search the requested item with the item found.
\end{definition}

{\bf Cost and complexity.} The cost metric in the search procedure can be used to model
diverse criteria. For example in~\cite{ORSW84} the cost corresponds to computational complexity;
underestimation costs one processor operation and overestimation two processor operations.
The cost can also correspond to other metrics like the budget in the airbags example
or some efficiency criterion in the TCP flow example.

\section{Sequences}
\label{sec:Sequences}

In the analysis of the lower bound and the search algorithm for the $(\alpha,\beta)$ search problem
we will use a generalization of Fibonacci sequences. In this Section we provide the
necessary definitions of the sequences that we will use.
We start with the well-known Fibonacci sequence.

\begin{definition}
The Fibonacci sequence is given by the recurrence relation
\begin{equation}
f(k) = f(k-1) + f(k-2) \quad ,
\end{equation}
with initial values $f(k) = 0$, for $k \leq 0$, and $f(1) = 1$.
\end{definition}

\noindent
We define the following generalization of the Fibonacci sequence
where each term is the sum of two preceding terms, which however may
not be the immediately preceding terms.

\begin{definition}
\label{equ:g()}
The $(\alpha,\beta)$ Fibonacci sequence.
Given integers $\alpha \geq 1$ and $\beta \geq 1$, the $(\alpha,\beta)$
Fibonacci sequence is given by the recurrence relation
\begin{equation}
g(k) = g(k-\alpha) + g(k-\beta) \quad ,
\end{equation}
with initial values, unless otherwise specified, $g(k) = 0$, for $k < 0$, and $g(0) = 1$.
\end{definition}

\noindent
The $(\alpha,\beta)$ Fibonacci sequence is actually a family of generalized
Fibonacci sequences. Some well known sequences that are related to
members of the $(\alpha,\beta)$ Fibonacci family are:
\begin{itemize}
\renewcommand{\labelitemi}{$\bullet$}
\item The classic Fibonacci sequence for $\{\alpha,\beta\} = \{1,2\}$ and initial values $f(0) = 0$ and $f(1) = 1$.
\item The Padovan sequence and the Perrin sequence for $\{\alpha,\beta\} = \{2,3\}$
and initial values $f(0) = f(1) = f(2) = 1$ and $f(0) = 3, f(1) = 0, f(2) = 2$, respectively.
\end{itemize}

\noindent
Given the parameters $\alpha$ and $\beta$ of an $(\alpha,\beta)$
Fibonacci sequence, let
\begin{equation}
\ell = \min\{\alpha,\beta\} \mbox{ and } u = \max\{\alpha,\beta\} \quad .
\end{equation}

\noindent
In the analysis of the $(\alpha,\beta)$ binary search problem we will encounter
the following sequence, in which the term $k$ of the sequence is equal to the sum of the
last $\ell$ terms of the corresponding $(\alpha,\beta)$
Fibonacci sequence.

%\begin{equation}
%G(k) = \sum_{i=1}^{\ell} g(k-i) \quad .
%\end{equation}

\begin{definition}
Given integers $\alpha \geq 1$ and $\beta \geq 1$, let $\ell = \min\{\alpha,\beta\}$.
Then let
\begin{equation}
\label{equ:G()}
G(k) = \sum_{i=1}^{\ell} g(k+1-i) \quad ,
\end{equation}
where $g(k)$ is the $(\alpha,\beta)$ Fibonacci sequence.
\end{definition}

\noindent
It turns out that the $G(k)$ sequence is also an $(\alpha,\beta)$ Fibonacci
sequence but with its own initial values.

\begin{proposition}
$G(k)$ is an $(\alpha,\beta)$ Fibonacci sequence.
\end{proposition}
\begin{proof}
\begin{equation}
\nonumber
G(k) = \sum_{i=1}^{\ell} g(k-i+1) = \sum_{i=1}^{\ell} g(k-i+1 - \alpha) + \sum_{i=1}^{\ell} g(k-i+1 - \beta)
\end{equation}
\begin{equation}
\nonumber
\Rightarrow G(k) = G(k-\alpha) + G(k-\beta) \quad .
\end{equation}
The appropriate initial values for $G(k)$ are $G(k) = 0$ for $k < 0$ and $G(k) = 1$ for $0 \leq k < u$.
\end{proof}

\noindent
Note that if $\ell = 1$ then, as expected, the sequences $G(k)$ and $g(k)$ coincide
(for the same parameters $\alpha$ and $\beta$).

Closed form expressions are known for example for the terms of the classic Fibonacci sequence,
the Padovan sequence and the Perrin sequence. Since every member of the $(\alpha,\beta)$ Fibonacci
sequences is defined by a linear recurrence, there exists a closed-form solution for each sequence.
However, we are not aware of a general closed form expression for the complete $(\alpha,\beta)$ Fibonacci family.
%Off course it would be interesting to come up with a closed form expression for the performance
%of the $(\alpha,\beta)$ Fibonacci search algorithm, but this presupposes general
%closed forms for $(\alpha,\beta)$ Fibonacci sequences which may not be easy to attain.
%There is however, an approximation of the terms of $(\alpha,\beta)$ Fibonacci sequences known
%from~\cite{KR89,Hi90} and previous work referenced therein. We will use this approximation in XXXX
%In this context, Johannes Kepler already observed that the ratio of consecutive Fibonacci
%numbers $(f(k+1)/f(k))$ converges to the golden ratio $\phi \approx 1.6180340$~\cite{wikipedia:Fibonacci_number}.
%This gives that the solution of $(1,2)$ or $(2,1)$ Fibonacci search has worst-case
%cost $O(\log_\phi n)$, while normal binary search $O(2 \cdot \log_2 n)$.
%For the Padovan and the Perrin sequences, despite their different initial values, the increase ratio of
%both sequences converges to the plastic number $p \approx 1.3247179$.
%Thus, $(1,3)$ or $(3,1)$ Fibonacci search gives solutions of worst-case
%cost $O(\log_p n)$, while normal binary search $O(3 \cdot \log_2 n)$.
Closed form estimations of terms of generalized Fibonacci sequences are given for
example in~\cite{CG01,Wo98,Hi90,KR89}. The following Proposition
follows directly from Lemma~1 of~\cite{Hi90} and provides upper and lower bounds
on the terms of the $G(k)$ Fibonacci sequence.
%\footnote{Johannes Kepler already observed that the ratio of consecutive Fibonacci
%numbers $(f(k+1)/f(k))$ converges to the golden ratio $\phi \approx 1.6180340$~\cite{wikipedia:Fibonacci_number}.
%This gives that the cost of $(1,2)$ or $(2,1)$ Fibonacci tree has worst-case
%cost $O(\log_\phi n)$, while normal binary search $O(2 \cdot \log_2 n)$.
%For the Padovan and the Perrin sequences, despite their different initial values, the increase ratio of
%both sequences converges to the plastic number $p \approx 1.3247179$.
%Thus, $(1,3)$ or $(3,1)$ Fibonacci search gives solutions of worst-case
%cost $O(\log_p n)$, while normal binary search $O(3 \cdot \log_2 n)$}.

\begin{proposition}
\label{prop:approxG()}
For strictly positive integers $\alpha$ and $\beta$ it holds
\begin{equation}
\label{equ:approxG()}
2^{\frac{k+1-u}{u}} \leq G(k) \leq 2^{\frac{k+2-u}{\ell}} \quad .
\end{equation}
\end{proposition}

\section{The Lower Bound}
\label{sec:bound}
We present an information-theoretic lower bound on the worst-case cost
of the $(\alpha,\beta)$ binary search problem. Like the known lower bounds
for binary search, the approach is based on decision trees. To handle
the asymmetry of the comparison costs we define a special type of decision
tree.

\begin{definition}
The $(\alpha,\beta)$ decision tree. The root node of the tree has level $0$.
Let $v$ be a node at level $d_v$ of the tree. Then\\
%\begin{center}
\begin{math}
\mbox{ node } v \left\{
\begin{array}{ll}
\mbox{ is either a leaf of the tree, } \\
\mbox{ or it has a left child at level $(d_v+\alpha)$ and a right child at level $(d_v+\beta)$. }
\end{array}
\right .
\end{math}
%\end{center}
\end{definition}

\noindent
Note that the term level is used to represent the weighted depth of the tree nodes.
Each level of the tree corresponds to a particular cost value.

\begin{definition}
The level of an $(\alpha,\beta)$ decision tree is the maximum level
of any of its nodes.
The depth of a node is the common depth of tree nodes, i.e., the number
of links from the root to the node. The depth of an $(\alpha,\beta)$ decision
tree is the maximum depth of any of its nodes.
\end{definition}

\noindent
The distinguishing property of $(\alpha,\beta)$ decision trees is
that nodes which have the same parent may be located at different
tree levels. The general form of an $(\alpha,\beta)$ decision tree is
shown in Figure~\ref{fig:A-BTree}. Two instances of $(\alpha,\beta)$ decision trees
are presented in Figure~\ref{fig:abFibTree}.

%\begin{figure}[!h]
%\begin{minipage}[b]{\linewidth}
%\centering
% %\includegraphics[width=\textwidth]{SimpleTrade.eps}
% \includegraphics[width=0.2\textwidth]{figA-BTree.eps}
%\caption{The $(\alpha,\beta)$ decision tree} \label{fig:A-BTree}
%\end{minipage}
%\end{figure}

\begin{figure}[!ht]
\centering
\subfloat[An $(\alpha,\beta)$ decision tree.]
%{\label{fig:hyb1}\includegraphics[viewport=40 400 760 750,width=7.3cm,height=4cm, clip]{figureHyb1.eps}}
{\label{fig:A-BTree}\includegraphics[width=0.35\textwidth, height=0.25\textwidth]{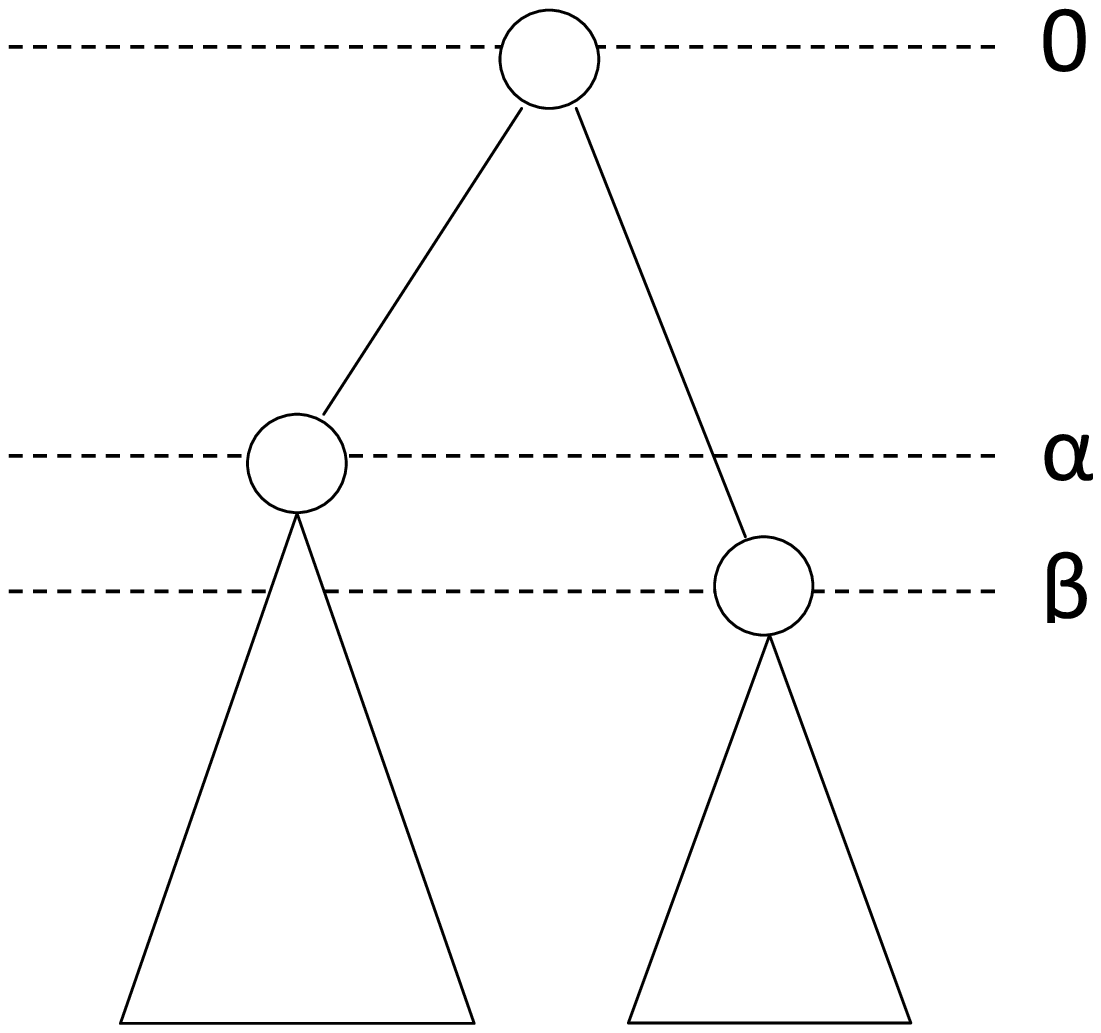}}
\hspace{0.1 \textwidth}
\subfloat[An h-decision tree for h=3.]
%{\label{fig:hyb2}\includegraphics[viewport=40 400 760 750,width=7.3cm,height=4cm, clip]{figureHyb2.eps}}
{\label{fig:HTree}\includegraphics[width=0.35\textwidth, height=0.25\textwidth]{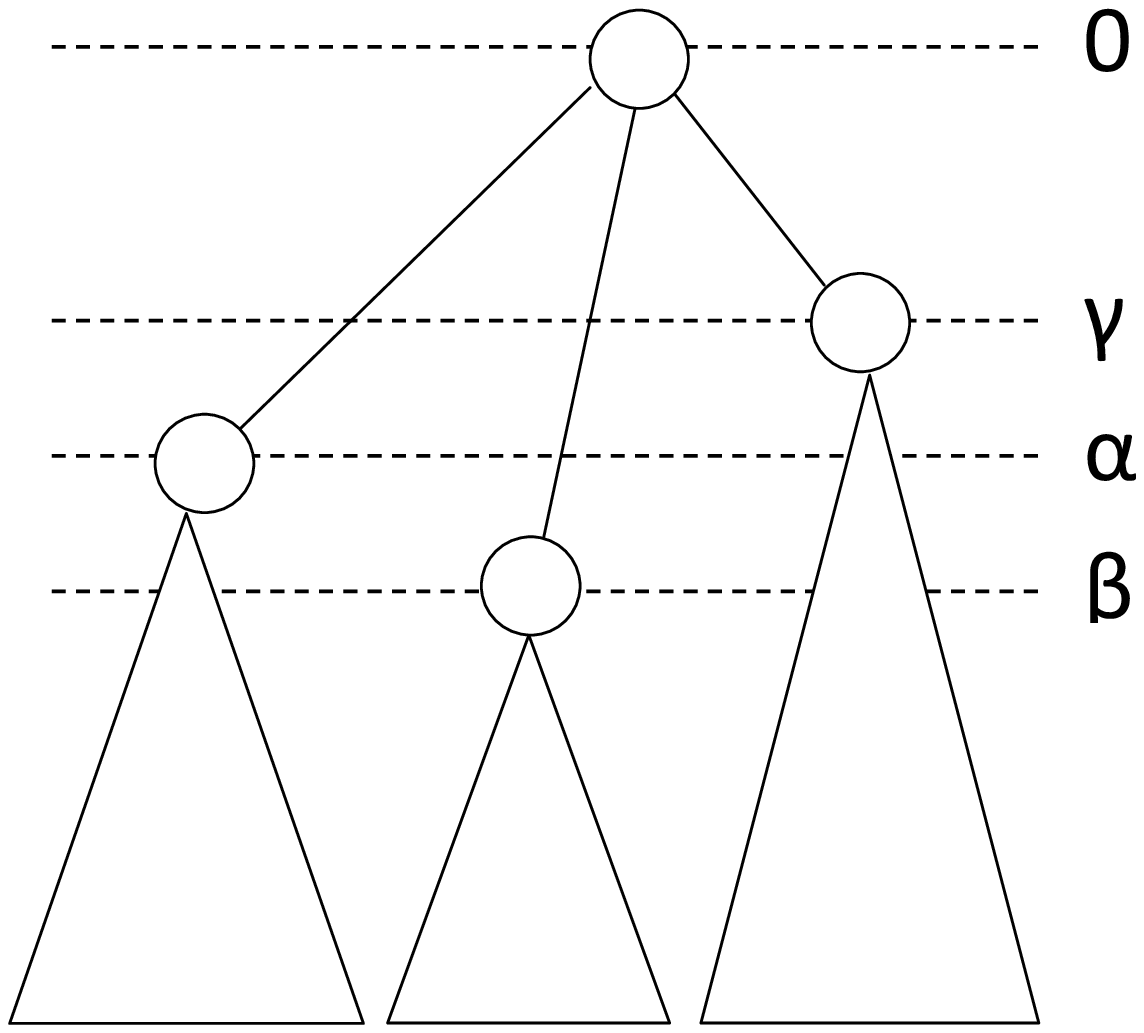}}
\caption{Two instances of asymmetric cost decision trees.} \label{fig:DecisionTrees}
\end{figure}

\begin{proposition}
The number of nodes at level $k$ of a complete $(\alpha,\beta)$ decision tree
is the value of term $g(k)$ of corresponding $(\alpha,\beta)$ Fibonacci sequence
of Equation~\ref{equ:g()}.
\end{proposition}
\begin{proof}
The number of nodes at level $k$ is equal to the sum of the number of
nodes at levels $(k-\alpha)$ and $(k-\beta)$. Furthermore, there
is one node at level $0$, the root node.
\end{proof}

\begin{proposition}
\label{prop:G-bound}
The number of leaf nodes of a complete $(\alpha,\beta)$ decision tree of level
$k$ is $G(k)$, where $G(k)$ is the $(\alpha,\beta)$ Fibonacci sequence of
Equation~\ref{equ:G()}.
\end{proposition}
\begin{proof}
The nodes of level $k$ are of course leaf nodes of the tree. Additionally,
any node with level in $(k-1), (k-2), \dots, (k+1-\ell)$ must also be
a leaf node since the node cannot have any descendants with level at most $k$.
\end{proof}

%\begin{figure}[!h]
%\begin{minipage}[b]{\linewidth}
%\centering
% %\includegraphics[width=\textwidth]{SimpleTrade.eps}
% \includegraphics[width=\textwidth]{figA-BTree3-1.eps}
%\caption{The $(3,1)$ decision tree with level 10. The maximum number of leaves
%of such a tree is equal to the number of nodes at level 10.
%The dark nodes are nodes where the search ends.} \label{fig:A-BTree3-1}
%\end{minipage}
%\end{figure}
%
%\begin{figure}[!h]
%\begin{minipage}[b]{\linewidth}
%\centering
% %\includegraphics[width=\textwidth]{SimpleTrade.eps}
% \includegraphics[width=\textwidth]{figA-BTree3-2.eps}
%\caption{The $(3,2)$ decision tree with level 10.
%The maximum number of leaves of such a tree is equal to the number of nodes at
%levels 9 and 10. The dark nodes are nodes where the search ends.} \label{fig:A-BTree3-2}
%\end{minipage}
%\end{figure}

\begin{proposition}
The wort-case cost for an $(\alpha,\beta)$ binary search problem with $n$ items
is at least $k$, where $k$ is the minimum index such that $G(k) \geq n$.
%term of the $(\alpha,\beta)$ Fibonacci sequence
\end{proposition}
\begin{proof}
Any search procedure for an $(\alpha,\beta)$ binary search problem
can be represented with a corresponding $(\alpha,\beta)$ decision tree.
The level of the decision tree
corresponds to the worst-case cost of the respective search procedure.
From Proposition~\ref{prop:G-bound} we know that any  $(\alpha,\beta)$
decision tree with level up to $k-1$ can have at most $G(k-1)$ leaves.
Hence any search procedure will have worst-case cost at least $k$,
where $k$ is such that $G(k) \geq n$.
\end{proof}

\section{$(\alpha,\beta)$ Fibonacci Search}
\label{sec:alg}

We propose $(\alpha,\beta)$ Fibonacci search, an algorithm for the
$(\alpha,\beta)$ binary search problem.
The proposed algorithm is a dichotomic search algorithm like binary search but
uses the $G(k)$ sequence numbers to dichotomize the item sets and to select
the items to be examined. The algorithm is a generalization of Fibonacci search~\cite{Fe60}
to $(\alpha,\beta)$ Fibonacci sequences.

\subsection{The Algorithm (short version)}
\label{sec:alg-short}
\begin{description}
\item[Algorithm $(\alpha,\beta)$ Fibonacci search] \item [Input:]
 A sorted array $V$ with $n$ items and a requested item $x$ in $V$.
 The items in V are indexed from $0$ to $n-1$.
\item [Step 0:] {\em $(\alpha,\beta)$ Fibonacci numbers.} Prepare the $(\alpha,\beta)$ Fibonacci numbers up to index $k$ such that $G(k) \geq n$.
 \item [Step 1:] {\em Initializations.}
 Let $z = k - \alpha$, $left = 0$, $right = n - 1$.
\item [Step 2:] {\em Search Loop.} \\
 While $(left < right)$
  \begin{enumerate}[(a)]
    \item $index = left + G(z);$ if $(index > right)$ then $\{index = right\};$
    \item $value = v[index]$
    \item Compare $(x \leq value)$
    \begin{description}
       \item [true:] $right = index; z = z - \alpha;$
       \item [false:] $left = index + 1; z = z - \beta;$ % if $(left > right)$ then $\{left = right\};$
       \end{description}
    \end{enumerate}
 \item [Step 3:] Check if $(V[index] == x).$ {\em This step is only necessary if the requested item x may not belong to V.}
 \end{description}

\noindent
First the algorithm prepares the sequence of $G(k)$ numbers up to
the first term $G(k)$ of the sequence such that $G(k) \geq n$.
Then, the algorithm enters the search loop. In each round, the algorithm
selects either the left subtree or the right subtree of the implicit
$(\alpha,\beta)$ decision tree. If the current subtree has $G(z)$ leaves
then the number of leaves of the left tree is $G(z-\alpha)$ and of the
right subtree $G(z-\beta)$. If the current subtree has less than $G(z)$ leaves
(if $n$ is not an $(\alpha,\beta)$ Fibonacci number),
then the right subtree will also have less than $G(z-\beta)$ leaves.

%This operation requires
%logarithmic time and space and has to be executed once for a particular search problem.
%
%Let $k$ be the argmin ....
%$g(k) = g(k-\alpha) + g(k-\beta)$.
%
%Select as ``middle'' item, the item at position left + $g(k-\alpha)$.
%Compare and select the appropriate part, etc.
%
%At each moment the algorithm has a sequence of leaves.
%The algorithm continues the sequence contains only one leaf (left and
%right bound of the sequence coincide).
%The algorithm may stop earlier than the last level, if the node has only
%one child. This improves the average cost. The worst-case cost is
%equal to the lowest level.

%Let $\ell = \min\{\alpha,\beta\}$. We distinguish two cases:
%\begin{description}
%\item [The Simple Case:] $\ell = 1$.
%In this case we can focus on the nodes of the last layer. For each node in layer
%$k$ there is a corresponding node in layer $k+1$.
%\item [The General Case:] $\ell \geq 1$. In this case, the nodes with cost at most
%$k$ are the nodes of the last $\ell$ levels.
%\end{description}

%\subsection{The Simple Case}

\begin{figure}[!ht]
\centering
\subfloat[The complete $(3,1)$ decision tree with level 10.
The maximum number of leaves of a $(3,1)$ decision tree with level 10 is equal to the number of nodes at
level 10 of the depicted tree.]
%{\label{fig:hyb1}\includegraphics[viewport=40 400 760 750,width=7.3cm,height=4cm, clip]{figureHyb1.eps}}
{\label{fig:abFib31}\includegraphics[width=0.9\textwidth]{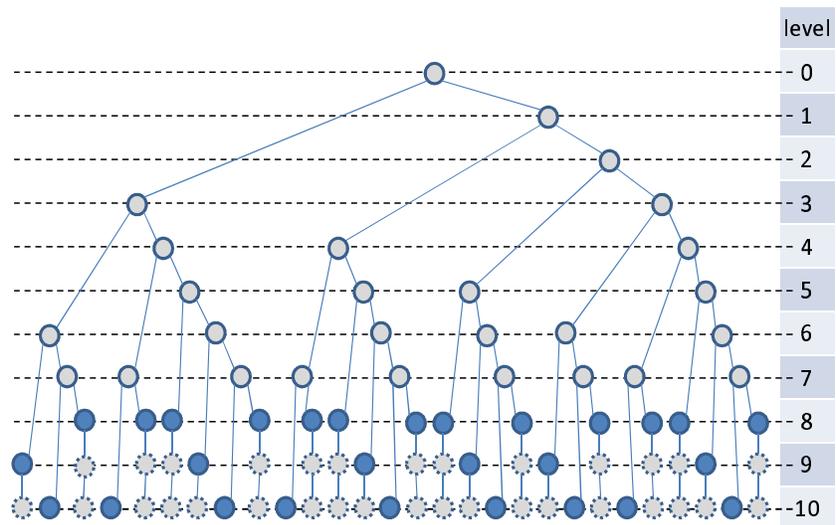}}
%\hspace{0.02 \textwidth}

\subfloat[The complete $(3,2)$ decision tree with level 10.
The maximum number of leaves of a $(3,2)$ decision tree with level 10 is equal to the number of nodes at
levels 9 and 10 of the depicted tree.]
%{\label{fig:hyb2}\includegraphics[viewport=40 400 760 750,width=7.3cm,height=4cm, clip]{figureHyb2.eps}}
{\label{fig:abFib32}\includegraphics[width=0.9\textwidth]{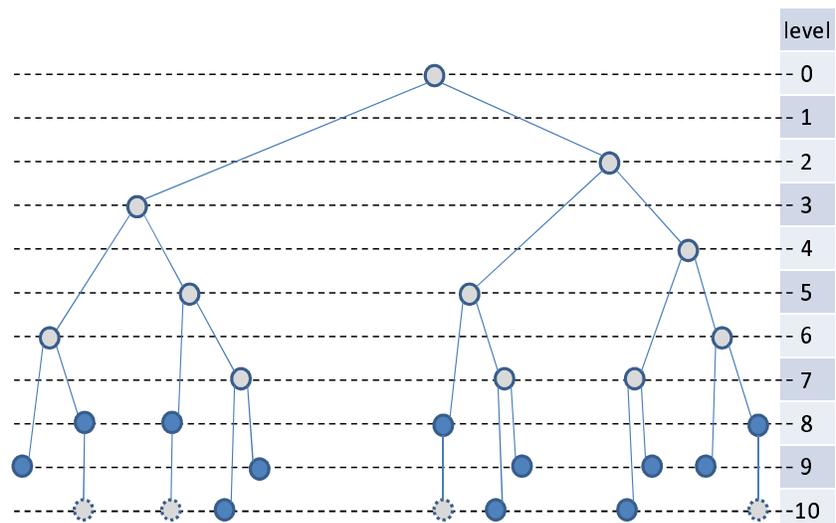}}
\caption{Two instances of $(\alpha,\beta)$ Fibonacci trees with level 10.
The dark nodes are the nodes at which the $(\alpha,\beta)$ Fibonacci search algorithm terminates.} \label{fig:abFibTree}
\end{figure}

\subsection{Complexity and Cost}

From the description of the $(\alpha,\beta)$ Fibonacci search algorithm
and the lower bound of Section~\ref{sec:bound}, it follows that the
worst-case cost of the algorithm is optimal.

\begin{proposition}
The worst-case cost of the $(\alpha,\beta)$ Fibonacci search algorithm
for the $(\alpha,\beta)$ binary search problem is optimal.
\end{proposition}
\begin{proof}
The decision tree of the search algorithm is the $(\alpha,\beta)$ decision
tree used in the proof of the lower bound in Section~\ref{sec:bound}.
Hence, the worst-case performance matches the lower bound.
\end{proof}

\noindent
We will show that the computational complexity of the $(\alpha,\beta)$ Fibonacci
search algorithm is logarithmic on the number $n$ of the items.

\begin{proposition}
\label{prop:abFibLevel}
Given an $(\alpha,\beta)$ binary search problem, the level of the implicit
$(\alpha,\beta)$ decision tree of algorithm $(\alpha,\beta)$ Fibonacci search
is at most $u \cdot \lceil\log_2 n\rceil$.
\end{proposition}
\begin{proof}
(Sketch.)
Recall that $u = \max \{ \alpha, \beta \}$.
The worst-case complexity of binary search for the $(\alpha,\beta)$ binary search problem
is $\lceil\log_2 n\rceil = O(\log n)$. Hence, the worst-case cost of binary search is at most
$u \cdot \lceil\log_2 n\rceil$. The worst-case cost of $(\alpha,\beta)$ Fibonacci search is
optimal and thus not larger that $u \cdot \lceil\log_2 n\rceil$.
\end{proof}

\begin{corollary}
\label{cor:depthABTree}
The depth of the $(\alpha,\beta)$ decision tree is not larger than $ \lceil u/\ell \rceil \lceil\log_2 n\rceil$.
\end{corollary}
\begin{proof}
Follows from Proposition~\ref{prop:abFibLevel} since every round of the search procedure increases
the level (cost) of the search by at least $\ell$. A slightly weaker result follows from
Proposition~\ref{prop:approxG()}.
\end{proof}

\begin{proposition}
The complexity of the $(\alpha,\beta)$ Fibonacci search algorithm is at most logarithmic
on the number of items $n$.
\end{proposition}
\begin{proof}
The algorithm first prepares the $(\alpha,\beta)$ Fibonacci sequence up to the first $k$,
such that $G(k) \geq n$. From Proposition~\ref{prop:abFibLevel}, the level of the corresponding
decision tree is bounded by $u \cdot \lceil\log_2 n\rceil$, and, thus, by $O(\log n)$, assuming
$u$ is a constant. From Corollary~\ref{cor:depthABTree} we obtain that
$k \leq \lceil u/\ell \rceil \cdot \lceil\log_2 n\rceil = O( \log n )$.
The main phase of the algorithm is the search loop. The loop is executed not more than
$\lceil u/\ell \rceil \cdot \lceil\log_2 n\rceil$ times, since in each loop the depth
of the remaining tree is reduced by at least one.
Thus, the overall complexity is $O(\lceil u/\ell \rceil \cdot \lceil\log_2 n\rceil) = O(\log n)$.
Note that a factor $u$ is hidden in the preparation time of the algorithm (calculation
of the terms of $G(k)$) and a factor $u/\ell$ in the complexity of each item search.
\end{proof}

\subsection{The Algorithm (full version)}
We present now the full version of the $(\alpha,\beta)$ Fibonacci search algorithm
which is also average-case optimal, assuming that all items of the population are equiprobable.
When $n = |V|$ is an $(\alpha,\beta)$ Fibonacci number, then the lowest layer of the
corresponding search tree is complete and the short version of Section~\ref{sec:alg-short}
is optimal in the average case too. If, however, $n$ is not an $(\alpha,\beta)$ Fibonacci number, then
some nodes of the lowest layer of the complete search tree are not used. To make the search
optimal in the average case, the surplus nodes have to be assigned to highest cost nodes
of the search tree.
Let $k$ be the level of the corresponding search tree. Then the number of surplus nodes
cannot exceed $g(k)$ because else the lowest of the tree could be avoided.
We assume that all surplus nodes are assigned to the lowest layer from right to left.
The full version of $(\alpha,\beta)$ Fibonacci search traverses the implicit decision
tree down to the corresponding leaf. During its descendance, it simply keeps track how
many surplus nodes correspond to its current subtree. The code of the algorithm becomes
a little more involved, but its asymptotic complexity remains the same. Due to lack space,
we omit the detailed description of the algorithm.

%\section{Engineering}
\section{Applications}
\label{sec:Apps}
Fibonacci sequences can be generalized to sequences of higher order.
For example, the common Fibonacci sequence of order three is the Tribonacci sequence.
Further members in increasing order are the Tetranacci, Pentanacci and Hexanacci sequences~\cite{wikipedia:Generalizations_of_Fibonacci_numbers}.
Since $(\alpha,\beta)$ Fibonacci search fits perfectly the decision tree of the $(\alpha,\beta)$
binary search problem we are tempted to examine what happens if we turn to higher order
Fibonacci sequences.

\subsection{A Guessing Game}
\label{subsec:game}
The generalization of the asymmetric cost binary search problem to searching with more outcomes
per comparison gives the following guessing game.

\begin{definition}
The h-outcomes guessing game. An array $V$ with $n$ items sorted in increasing order and access time
$O(1)$ for each item, is given. To search in an interval I, the following comparison operation
is supported: The interval I has to be partitioned into a partition P of h consecutive intervals defined
with h-1 values $v_1 \leq v_2 \leq \dots \leq v_{h-1}$. For any value $x$ and any partition P,
the cost of finding in which of the intervals of P item $x$ belongs, is
\begin{center}
\begin{math}
\mbox{cost of (in which interval of P is x)} = w_i, \mbox{ if the searched item is in interval i, for } i = 0,1,\dots, h-1,
\end{math}
\end{center}
where the weights or costs $w_i$, for $i = 0,1, \dots, h-1$, are strictly positive integers.
The cost of a search is equal to the sum of the costs of all comparison
operations performed during the search. As in the binary search case, the requested
item $x$ is assumed to belong to the array $V,$ else the requested item has to be compared
at the end of the search with the item found.
\end{definition}

The results of Section~\ref{sec:bound} on the lower bound of the two outcomes $(\alpha,\beta)$
binary search can be extended to obtain h-decision trees this search problem with $h$ outcomes.
Similarly, it is straightforward to generalize the $(\alpha,\beta)$ Fibonacci search algorithm
of Section~\ref{sec:alg} to the h-Fibonacci search algorithm that optimally solves the h-outcomes
search problem. Due to lack space, we omit the detailed description and proofs of these results.

\subsection{Coding with Unequal Letter Costs}
\label{subsec:coding}

The problem of searching with asymmetric costs is strongly related to
optimal prefix codes~\cite{Va71,GKY02}. The encoding of letters that
may have different lengths but occur with the same probability is
known as the Varn coding problem and can be solved in polynomial time.

The results of this work, provide simple strict lower bounds
on the worst-case cost of Varn codes with lengths that are integers
or rational multiples of each other.
In particular, an implicit description of an optimal Varn encoding or
an optimal Varn encoding with the alphabetic property (the additional
constraint that the encoding of the words must preserve their alphabetical
order) can be calculated in logarithmic time.
The h-Fibonacci decision tree provides a tight lower bound on the worst case
cost for the encoded form and the h-Fibonacci search algorithm gives an optimal
coding of the letters that matches the lower bound.
Then, the encoding of any particular word can also be calculated
in logarithmic time. This might be useful if the domain (the item set
to be encoded) is very large, for example exponential on some problem parameter.
In this case we may generate an implicit encoding, given any specific
(possibly implicit) ordering of all words and then generate only the words
that are actually used.

%If the word have the same probability (Varn coding) ...
%
%An encoding can be calculated in XXXXX time and implicitly described.
%
%Given any order of the words to be encoded ...

\section{Discussion}
\label{sec:disc}
We consider the $(\alpha,\beta)$ Fibonacci search algorithm and its generalization,
the h-Fibonacci search algorithm, fundamental tools that can be used within the
solutions of other computational problems. For this reason, we developed a
prototype implementation of the two algorithms within a standard Java class library
that can be found at~\cite{abFib}.
%Design of an abstract search algorithm ABFibSearch and implementation of a corresponding Java class.

Back to the airbag example of Section~\ref{sec:intro}, the interval of uncertainty
contains $101$ distinct values $(10.0,10.1,\dots,20.0)$. We use the implemented
$(\alpha,\beta)$ Fibonacci search algorithm with $\alpha = 1$ and $\beta = 3$
to find that the level of the corresponding $(\alpha,\beta)$ decision tree is $14$
(since G(13) = 88 and G(14) = 129, for $\alpha=1$ and $\beta = 3$). Thus,
the worst-case cost is $500 \cdot 14 = $ {7000\euro}. The corresponding tree
generated by normal binary search has worst-case depth $\lceil \log_2 101 \rceil = 7$
and thus, the cost of the binary search solution can be up to $7 \cdot 1500 =${10500\euro}.
Since the lowest layer of the binary search tree is not complete, this cost can be reduced
to $6 \cdot 1500 + 1 \cdot 500 = $ {9500\euro} by using the $101$ leftmost leaves (because
in this case the right branches have the higher cost).
%XXXXX A similar trick can be used in the $(\alpha,\beta)$ Fibonacci search algorithm too,
%in this case to (possibly) improve the average behavior of the algorithm (since the worst-case
%performance is already optimal).

The difference between the costs $7000$ and $9500$, albeit not colossal, is
without doubt significant. Preliminary experiments indicate the, rather expected, result,
that the benefits of using $(\alpha,\beta)$ Fibonacci search increase as the
ratio $u/\ell$ increases. The experiments also indicate that the benefits in the average
case cost of the search problem if $(\alpha,\beta)$ Fibonacci search is used, are still
significant, albeit smaller than with the worst-case criterion.

A further interesting question is if the results of this work can be generalized
to support real positives $\alpha$ and $\beta$, possibly with the use of the Pascal's
triangle like in~\cite{CW77} or techniques from~\cite{CG01}.

\end{document}